\documentclass{article}
\usepackage[utf8]{inputenc}
\usepackage{slashed}
\usepackage{amsmath}
\usepackage{amsfonts}
\usepackage{amsthm}
\newtheorem{theorem}{Theorem}

\title{Integrating Factors for Dirac-Schrodinger Operators: Improving Eigenvalue Estimates and Applications to Charged Positive Mass Theorems Outside Horizon(s)}
\author{Robert Abramovic}
\date{February 2019}

\begin{document}

\maketitle

\section{Introduction}

\begin{abstract} 
Let $(M^{n}, g)$ denote a Riemannian spin manifold of dimension $n$ with Dirac operator $\slashed{D}$ induced from the Levi-Cevita connection acing on the spinor bundle, $S$ ($\slashed{D}$ is also called the Atiyah-Singer Operator). Let $c: Cl(TM^{n}) \rightarrow End(S)$ be the standard representation of the Clifford Algebra as endomorphisms of the spinor bundle. Let $B \in End(S)$ be a zeroth-order endomorphism of the spinor bundle; given an in an orthonormal frame, $e_{j} \in TM^{n}$ by the expression $B=f^{\alpha}c(e_{\alpha})$ where the sum is taken over multi-indices, $\alpha = (i_{j_{m}}), \ m = 1, \, 2, \, 3 \, ... ,\ k$, $j_{1}< j_{2} < ... < j_{k}$ and each $f^{\alpha} \in C^{\infty}(M^{n})$. The purpose of this paper is investigate when the Dirac-Schrodinger operator $\slashed{D} + B$ has an integrating factor, i.e. when does there exist an invertible endomorphism $A \in End(S)$ such that $\slashed{D}(AB)=A\slashed{D}+AB$. This has applications to improving eigenvalue estimates for Dirac-Schrodinger operators and proving positive charged positive mass theorems where such operators appear on the boundary. Of particular interest is the case $n = 2$, for boundary Dirac operators of this form appear in charged positive mass theorems based on the initial data formulation in mathematical general relativity.  It allows us to generalize a theorem of M. Herzlich (set-forth in his attempt to prove the Riemannian Penrose-inequality using spinors, cf. [1]) to a manifold of dimension $n \geq 3$ containing an electric field and symmetric two-tensor representing the second-fundamental form. 
\end{abstract} 

\pagebreak

\subsection{Theorems on The Existence of Integrating Factors for Dirac-Schrodinger Operators} 

Here, we present a summary of important results about integrating factors for Dirac-Schrodinger Operators. Throughout, $(M^{n}, g)$ will be a Riemannian spin manifold with Riemannian metric $g$ and Levi-Cevita connection $\nabla$. Recall that $M^{n}$ has a spinor bundle, $S$, that inherits a connection, also denoted by $\nabla$, from the Levi-Cevita connection. Let $\{ e_{j} \}$ denote an orthonormal frame for $TM^{n}$. The \emph{classical Dirac operator} $\slashed{D}$ acting on sections of $S$, will be defined by the following equation: 
\begin{equation} 
\slashed{D} = c(e_{j})\nabla_{e_{j}}
\end{equation} 
where there is an implied sum from $j=1$ to $j=n$. If $B \in End(S)$ is given by $B=f^{\alpha}c(e_{\alpha})$, then $\slashed{D}+B$ is called a \emph{Dirac-Schrodinger} operator. Now, note that the requirement that $A$ be an integrating factor for $\slashed{D}+B$ reduces to the following equation for operators: 

\begin{equation} 
\slashed{D}A=AB
\end{equation} 

\noindent Since we require an endormorphism $A$ to be invertible, we can multiply both sides by $A^{-1}$ to obtain the following:

\begin{equation} 
A^{-1}\slashed{D}A = B
\end{equation} 

\noindent In dimension 1, the endomorphism $B$ is simply a smooth complex-valued function and we have the following simple theorem: 
\begin{theorem}

On $M^{1}$, the integrating factor for the Dirac-Schrodinger operator $\slashed{D}+B$ is the smooth complex-valued function, $A(x)$, given by the following equation: \\

\begin{equation} 
A(x) = \exp(-i\int B(x) dx)
\end{equation}
If $M^{1} = \mathbb{R}$, then $A$ will always exist. If $M^{1} = S^{1}$, then $A$ will exist if and only if $\int_{M^{1}} B = n$ for $n \in \mathbb{Z}$. 
\end{theorem} 

\noindent Before moving to the case in arbitrary dimension, we need to define the following operations for $X \in Gamma(TM^{n})$: \\
\\
Let $\{ e_{j} \}$ be an orthonormal frame for $TM^{n}$. Then we  define: \\
\\
\emph{The Divergence} 
\begin{equation} 
div \ X = \nabla_{e_{j}}X^{j}
\end{equation}
and \\
\\
\emph{The Curl} 
\begin{equation} 
curl \ X = \Sigma_{m<j} [\nabla_{e_{m}}X^{j}-\nabla_{e_{j}}X^{m}]c(e_{m})c(e_{j})
\end{equation} 

\noindent Note that the equation for the curl differs slightly from its usual definition in a vector calculus course in this sense that it is not a vector, but instead includes factors that are representations of bivectors.
\begin{theorem}
Let $(M^{n}, g)$ be a manifold of arbitrary dimension, $n$. Let $\slashed{D}+B$ be a Dirac-Schrodinger operator acting on the spinor bundle. Then $\slashed{D} + B$ admits an integrating factor if one of the two cases hold: \\
\\
\emph{Case 1}: If $B = c(\nabla f)$, then $\slashed{D}+B$ has the integrating factor $A = \exp(f)$. \\
\\
\\
\emph{Case 2:} Assume that $(M^{n}, g)$ admits a nonvanishing vector field $X \in TM^{n}$,  If $B = \frac{c(X)}{|X|^{2}}[(div \ X)-(curl \ X)]$, then $\slashed{D}+B$ has the integrating factor $A = c(X)$.
\end{theorem} 

\noindent \textbf{Remark:} In the rest of the paper, I will call the $B=c(\nabla f)$ \textbf{Form 1} and $B=\frac{c(X)}{|X|^{2}}[(div \ X) - (curl \ X)]$ \textbf{Form 2}. \\
\\
The above two cases relied simply on computations of $A^{-1}\slashed{D}A$ that will be carried out in section 3. However, there is a special case where $f \in C^{\infty}(TM^{n})$ is a function that depends only one variable in some orthonormal frame, i.e. in the orthonormal frame $\{ e_{j}, \ j=1,..., n \}$, $\nabla_{e_{j}} f=0$ for all $j \geq 2$. In this case, $\Delta_{g} f =-\Sigma_{j=1}^{n} \nabla_{e_{j}}\nabla_{e_{j}}f = -\nabla_{e_{1}}\nabla_{e_{1}}f$ and we will consider the case where $B=\Delta_{g}f$. As we shall prove in section 4, in this case, the integrating factor for $\slashed{D}+B$ can be written as 
\begin{equation}
A=exp(c(\nabla f))
\end{equation} where for any endomorphism $C \in End(S)$, $exp(C)$ is defined as follows: \\
\begin{equation}
exp(C)=\Sigma_{n=0}^{\infty} C^{n}
\end{equation}
The sum on the left-hand side of (8) is well-defined because $C$ has the local expression as a matrix acting on fibers of the spinor bundle $S$, which are themselves vector spaces. \\
\\
\begin{equation}
A = exp(c(\nabla f))    
\end{equation}
\\
We summarize this special result in the following theorem: \\
\\
\\
\begin{theorem} 
Let $(M^{n}, g)$ be a Riemannian spin manifold of arbitrary dimension $n$. $\{e_{j}, j = 1, ..., n\}$ will denote an orthonormal frame for $TM^{n}$. Let $f \in C^{\infty}(M^{n})$ be a smooth function such that $\nabla_{e_{j}} f =0$ for all $j \geq 2$. Define \\
\begin{equation} 
B=\Delta_{g}f 
\end{equation} 
Note that $B \in End(S)$ because the operator acts by scalar multiplication. Then the Dirac-Schrodinger operator $\slashed{D}+B$ has the following integrating factor: \\
\\
\begin{equation} 
A=exp(c\nabla f)
\end{equation}
\end{theorem} 
\noindent The formula for $B$ given in equation 10 will be called \textbf{Form 3}. \\
\\
We now turn to the two-dimensional case where $(M^{2}, g)$ is a surface. In this case $\{ e_{1}, e_{2} \}$ will be an orthonormal frame for $TM^{2}$. We will let $X \in \Gamma(TM^{2})$ be a vector field and $f \in C^{\infty}(M^{2})$ and $g \in C^{\infty}(M^{2})$ will be smooth functions; again, we will use the definition of divergence and curl given in equations (4) and (5). The most general form for $A \in End(S)$ is following: \\
\\
\begin{equation} 
A=f+c(X)+gc(e_{1}e_{2})
\end{equation}
\\
Provided that $f$, $g$, and $X \in \Gamma(TM^{2})$ do not simultaneously vanish on $M^{2}$, the inverse of this operator $A^{-1} \in End(S)$ is given by the following: \\
\\
\begin{equation} 
A^{-1}=\frac{f-c(X)-gc(e_{1}e_{2})}{f^{2}+g^{2}+|X|^{2}}
\end{equation}
\\
By simply computing $B=A^{-1}\slashed{D}A$, we obtain the following (the details are carred out in section 5): \\
\\
\begin{equation}
B=\frac{f-c(X)-gc(e_{1}e_{2})}{f^{2}+g^{2}+|X|^{2}}[-(div \ X)+[c(\nabla f)+(\nabla_{e_{2}}g)c(e_{1})-(\nabla_{e_{1}}g)c(e_{2})]+curl \ X]
\end{equation}
\\
The formula for $B$ given in equation 14 will be called $\textbf{Form 4}$. As result we have the following theorem: 
\begin{theorem} 
Let $(M^{2}, g)$ be a Riemannian manifold of dimension 2 (surface) and let $\{ e_{1}, \ e_{2} \}$ be a local orthonormal frame for $TM^{2}$. Let $f, \ g \in C^{\infty}(M^{2})$ and $X \in \Gamma(TM^{2})$ be given such that they never simultaneously vanish on $M^{2}$. If $B \in End(S)$ has form 4, then the Dirac-Schrodinger operator $\slashed{D}+B$ will have an integrating factor $A \in End(S)$ given by the following: \\
\\
\begin{equation}
A = f+c(X)+gc(e_{1}e_{2})
\end{equation}
\end{theorem}
\subsection{A Theorem on Eigenvalues  and Eigenvalue Estimates} 
The theorems on the existence of integrating factors given in section 1.1 allow us to develop and prove theorems about the eigenvalues of Dirac-Schrodinger operators. Throughout, we will let $\lambda(\slashed{D})$ denote an eigenvalue of $\slashed{D}$ and $\lambda(\slashed{D}+B)$ denote an eigenvalue of the Dirac-Schrodinger operator $\slashed{D}+B$. In particular, consider the eigenvalue equation \\
\\
\begin{equation} 
(\slashed{D}+B)\psi = \lambda \psi
\end{equation} 
for the eigenvalue $\lambda$. Note that if $\slashed{D}+B$ admits an integrating factor $A$, then we can multiply both both sides of equation (15) by $A$ to obtain $(A\slashed{D}+AB)\psi = \lambda (A\psi)$. Since $A$ is an integrating factor, we have $A\slashed{D}+AB=\slashed{D}(AB)$ and therefore we get the equation
\\
\begin{equation}
\slashed{D}(A\psi)=\lambda (A\psi)
\end{equation}
Since $A$, being an integrating factor, is invertible, we can conclude that $\lambda$ is an eigenvalue of $\slashed{D}$ if and only if it is an eigenvalue of $\slashed{D}+B$. \\
\\
In the special case that $(M^{n}, g)$ is compact without boundary, consider the Dirac-Schrodinger operator, $\slashed{D}+f$ for a function $f \in C^{\infty}(M^{n})$. Let $f_{avg} = \frac{1}{|M^{n}|}\int_{M^{n}} f$ denote the average of $f$ and note that by construction, the function $f-f_{avg}$ integrates to zero. In this case, the solutions of the Laplace equation on $M^{n}$ are simply constant functions and therefore, we can conclude by Fredholm's alternative that there exists an $h \in C^{\infty}(M^{n})$ with $\Delta_{g} h = f-f_{avg}$. Assume further that exists an orthnormal frame $\{ e_{j}, j = 1, ..., n \}$ for which $\nabla_{e_{j}}f=0$ for all $j \geq 2$. Then $h$ can be chosen to depend only on $e_{1}$, so that $f-f_{avg}=-\nabla_{e_{1}}\nabla_{e_{1}}h=\Delta_{g}h$ has the special form of $B$ in theorem 3, so that Dirac-Schrodinger operator $\slashed{D}+(f-f_{avg})$ admits the integrating factor $A=exp(c(\nabla h))$. We therefore can conclude that the eigenvalues of $\slashed{D}+f$ are precisely those of $\slashed{D}$ shifted by $f_{avg}$. If $M^{2}=S^{2}$, then an argument of Hijazi-Bar (cf. [2], [3]) gives on the following lower-bound on eigenvalues of the Dirac operator, $\slashed{D}$: $|\lambda(\slashed{D})| \geq \sqrt{\frac{4\pi}{|M^{2}|}}$, from whence we conclude that $|\lambda(\slashed{D}+f)| \geq \sqrt{\frac{4\pi}{|M^{2}|}}-|f_{avg}|$, provided that $\nabla_{e_{2}}f=0$. \\
\\
These results on eigenvalues prove the following: 
\begin{theorem}
Let $(M^{n}, g)$ be a dimension $n$ Riemannian spin manifold with spinor bundle $S$ and let $\slashed{D}+B$ be a Dirac-Schrodinger operator acting $\Gamma(S)$. If $\slashed{D}+B$ admits an integrating factor, $A$, then the set of eigenvalues of $\slashed{D}$ coincides with the set of eigenvalues of $\slashed{D}+B$. If $B=f$ with $\nabla_{e_{2}}f=0$ in some orthonormal frame $\{ e_{j}, j=1,...,n \}$ for $TM^{n}$, then the eigenvalues of the Dirac-Schrodinger operator $\slashed{D}+f$ are precisely the eigenvalues of $\slashed{D}$ shifted by $f_{avg}$, i.e. \\
\begin{equation}
\lambda(\slashed{D}+f)=\lambda(\slashed{D})+f_{avg}
\end{equation} 
If, in addition, $M^{2}=S^{2}$, then we obtain the following lower bound on the eigenvalues of $\slashed{D}=f$: \\
\\
\begin{equation}
|\lambda(\slashed{D}+f)|\geq \sqrt{\frac{4\pi}{|M^{2}|}}-|f_{avg}|
\end{equation}
\end{theorem}
\noindent Note that the lower bound given in (19) is an improvement on the lower bound for $\slashed{D}+f$ given in B. Morel's paper [3]. \\
\\
\subsection{Application to Positive Mass Theorems Outside Horizon(s)} 
In this subsection, we apply theorems 2-4 to prove positive mass theorems outside horizon(s) for initial data for Eintein's equation. Initial data consists for Einstein's field equations generally consists of a Riemannian manifold, $(M^{3}, g)$ in addition to a vector field $E \in \Gamma(TM^{3})$ representing the electric field, and a symmetric two-tensor $k \in \Lambda^{2}(TM^{3})$ defining how $M^{3}$ is embedded in a 4-dimensional spacetime. If $(M^{n}, g)$ is a Riemannian manifold of arbitrary dimension $n$, with electric field $E \in \Gamma(TM^{n})$, then we can also prove a charged positive mass theorem outside horizon(s), but in this case $(M^{n}, g, E)$ does not give initial data for Einstein's field equations for a spacetime. These theorems will begin with the assumption that $(M^{n}, g)$ is asymptotically flat, a definition that will be made precise later. On asymptotically flat manifolds, we can define ADM energy, $E_{ADM}$, which measures the strength of the strength of the gravitational field from spatial infinity. The precise definition of $E_{ADM}$ is given in section 7. If $(M^{n}, g)$ also has a vector field $E \in \Gamma(TM^{n})$, then we can define the total charge, $Q$. The following charged positive mass theorems are divided into three cases: initial data in a time-symmetric spacetime, initial data in a non-time-symmetric spacetime, and a manifold of arbitrary dimension $n$: \\
\\
\begin{theorem}{\textbf{(Time-Symmetric PMT Outside Horizon(s))}}
Let $(M^{3}, g, E)$ be charged initial data for Einstein's equations in a time-symmetric spacetime. In particular, $(M^{3},g)$ is asymptotically flat and $E \in \Gamma(TM^{3})$ with $div \ E =0$. Assume that $\partial M^{3} = \cup_{j} N_{j}$ with $N_{j} = S^{2}$ for each $j$. Let $\nu$ denote an outward normal vector field on each $N_{j}$; define $E(\nu) = g(E, \nu)$ and $Q_{j} = \frac{1}{4\pi} \int_{N_{j}} E(\nu)$. $H_{j}$ will denote the mean curvature on each $N_{j}$ and $A_{j} = |N_{j}|$ will denote the area of each sphere.  Let $R$ denote the scalar curvature of $M^{3}$ and assume that $R-2|E|_{g}^{2} \geq 0$. We then have the following: \\
\\
\\
\\
\textbf{Case 1 (my thesis [9])}: If $\int_{M^{2}} (E(\nu))^{2} \leq 4\pi + \frac{16\pi Q_{j}^{2}}{A_{j}}$ and  $\frac{1}{2}H_{j} \leq \sqrt{\frac{4\pi}{A_{j}} + \frac{16\pi^{2}Q_{j}^{2}}{A_{j}^{2}}}-\frac{4\pi |Q_{j}|}{A_{j}}$ for each $j$, then $E_{ADM} \geq |Q|$. \\
\\
\\
\textbf{Case 2}: If on each $N_{j}$, there is an orthonormal frame $\{ e_{1}, e_{2} \}$ such that $\nabla_{e_{2}} E(\nu) = 0$ and if $\frac{1}{2}H_{j} \leq \frac{\sqrt{4\pi}}{A_{j}} - \frac{4\pi|Q_{j}|}{A_{j}}$, then $E_{ADM} \geq |Q|$. 
\end{theorem} 
\noindent If we also include a symmetric two-tensor, $k \in \Lambda^{2}(TM^{3})$ representing how $M^{3}$ is embedded into the spacetime, then we can also define the matter density, $\mu$, the matter density $J \in \Gamma(TM^{3})$, and the ADM momentum, $P \in \Gamma(TM^{3})$ of $M^{3}$. Again, details will be given in section 6. \\
\begin{theorem}{\textbf{(Non-Time-Symmetric Charged PMT Outside Horizon(s))}}
Let $(M^{3}, g, E, k)$ be initial data for Einstein's equations in a non-time-symmetric spacetime. In particular, $(M^{3}, g)$ is asymptotically flat, $E \in \Gamma(TM^{3})$ with $div \ E=0$, and $k \in \Lambda^{2}(TM^{3})$ is symmetric, representing how $M^{3}$ is embedded into the spacetime. The same notation as the previous theorem will apply and we will set $e_{1} = \nu$ on $N_{j}$ and let $\{ e_{A} \}$ denote an orthonormal frame on each $N_{j}$. We will assume that $\mu \geq |J|_{g}$. Then we obtain the following: \\
\\
\textbf{Case 1}: If $B = E(\nu)+(\frac{1}{2}\Sigma_{m} (k_{mm}))c(e_{1})$ has form 4 $\frac{1}{2}H_{j} \leq \sqrt{\frac{4\pi}{A_{j}}}$ on each $N_{j}$, $E_{ADM} \geq \sqrt{Q^{2} + |P|_{g}^{2}}$. \\
\\
\textbf{Case 2}: If $B = E(\nu) - \frac{4\pi Q_{j}}{A_{j}}+(\frac{1}{2}\Sigma_{m} (k_{mm}))c(e_{1})$ has form 4 and $\frac{1}{2}H_{j} \leq \sqrt{\frac{4\pi}{A_{j}}} - \frac{4\pi |Q_{j}|}{A_{j}}$ on each $N_{j}$, then $E_{ADM} \geq \sqrt{|P|_{g}^{2}+Q^{2}}$. \\
\\
\textbf{Case 3 (special case)}: If on, each $N_{j}$, $E(\nu) = 0$, $(\Sigma_{m=2}^{3} k_{mm})c(e_{1})$ has form 1 or form 2 and $\frac{1}{2}H_{j} \leq \sqrt{\frac{4\pi}{A_{j}}}$, then $E_{ADM} \geq |P|_{g}$.

\end{theorem} 
\noindent \textbf{Remark:} In the preceding PMT theorems, we have seen forms 1, 3, and 4 for the Dirac-Schrodinger operator. Form 2 was ignored because there is no non-vanishing vector field on $S^{2}$ by the hairy ball theorem. However, if we move to dimension $n > 3$, then it is possible for there to be a non-vanishing vector field on a component of the boundary of $M^{n}$, denoted by $\partial M^{n}$.  In particular, we have the following general theorem: \\
\begin{theorem}{\textbf{(Higher-Dimensional Charged Non-Time Symmetric PMT Outside Horizon(s))}}
Let $(M^{n}, g)$ be an asymptotically flat Riemannian spin manifold of dimension $n > 3$. Let $E \in \Gamma(TM^{n})$ represent the electric field and let $k \in \Lambda^{2}(TM^{n})$ be a symmetric two-tensor. Assume that $\partial M^{n} = \cup_{j} N_{j}$, with each $N_{j}$ compact without boundary and use the same notation as the previous two theorems for $Q_{j}$ and $H_{j}$. As in the previous two theorems, $R$ will denote the scalar curvature of $M^{n}$ and we will assume that $R-2|E|_{g}^{2} \geq 0$. Further, we assume that $R_{j} > 0$ on all of $N_{j}$ on all of $N_{j}$ where $R_{j}$ denotes the scalar curvature of $N_{j}$ induced from the metric $g$. Let $\lambda_{j}$ denote the first positive eigenvalue of $\slashed{D}_{N_{j}}$, the classical Dirac (Atiyah-Singer) operator on $N_{j}$. Again, $e_{1}$ will denote an outward normal pointing vector field on each $N_{j}$ and $H_{j}$ will denote the mean curvature of $N_{j}$. If on each $N_{j}$,  $-\frac{1}{2} \Sigma_{m=1}^{n-1} k_{mm}c(e_{1}) + \frac{n-3}{2}(proj_{N_{j}} E)-\frac{n-1}{2}E^{1}$ has form 1, 2, or 3 and $\frac{1}{2}H_{j} \leq \lambda_{j}$, then $E_{ADM} \geq \sqrt{|P|_{g}^{2}+Q^{2}}$. If $[-\frac{1}{2} \Sigma_{m=1}^{n-1} k_{mm}c(e_{1}) + \frac{n-3}{2}(proj_{\partial M^{n}}E)-\frac{n-1}{2}E^{1}]-\frac{4\pi Q_{j}}{|N_{j}|}$ has form 1, 2, or 3 on each $N_{j}$ and $\frac{1}{2}H_{j} \leq \lambda_{j} -\frac{4\pi |Q_{j}|}{|N_{j}|}$, then $E_{ADM} \geq \sqrt{|P|_{g}^{2}+Q^{2}}$ 
\end{theorem} 
\noindent \textbf{Remark 1:} Notice that in dimension $n=2$, the requirement that $R_{j} > 0$ restricts each $N_{j}$ to being sphere, thus explaining the topological restriction placed in theorems 6 and 7. \\
\\
\noindent \textbf{Remark 2:} Note that every Riemannian 3-manifold is spin, so that this was not needed as an additional assumption in theorems 6 and 7. \\
\section{1 Dimensional Case: Proof of Theorems 1 and Discussion of Eigenvalues} 

\subsection{Harmonic-Type Spinor Fields}
If $n =1$, then $\slashed{D} = i\frac{d}{dx}$ and $B=fI$, where $f$ is a smooth complex-valued function since in this case $Cl(TM^{1})$ is isomorphic to $\mathbb{C}$. Therefore $\slashed{D} + B = i\frac{d}{dx} + f$. Further, $S$ is simply a complex line bundle over $M$ and its sections (spinor fields) $\psi \in \Gamma(S)$ can have a local expression as complex valued functions on which $\slashed{D} + B$ acts; $\slashed{D}$ acts by differentiation and $B$ acts by multiplication. We start investigating equations involving by this operator by first defining harmonic-type spinors. \\ 
\\
\textbf{Definition:} If $\psi \in \Gamma(S)$ solves the equation $(\slashed{D}+ B)\psi = 0$, then it will be called a \emph{harmonic-type} spinor field for $\slashed{D} + B$ \\
\\
This definition is motivated by the fact that if $\psi$ solves $\slashed{D}\psi = 0$, then it is called a \emph{harmonic} spinor.
\\
\\
\textbf{Remark:} The above equation $(\slashed{D} + B)\psi = 0$ appears in the positive mass theorem with charge and for the positive mass theorem in the non-time symmetric case, since in that case a new connection yields a Dirac-Schrodinger operator of the form $\slashed{D} + B$; this has important applications in mathematical general relativity if $n \geq 3$. \\ 
\\
\noindent Now, let $\psi$ be a harmonic-type spinor field for $\slashed{D} + B$, then $\psi$ solves the equation \\
\\
\begin{equation}
i\frac{d\psi}{dx}+ f\psi = 0
\end{equation} 
\\
If $f:M \rightarrow \mathbb{C}$ is a given function, then this is simply an ordinary differential equation (ODE), and it can be solved using the method of integrating factors. In particular, we can multiply both sides of equation (1) by $-i$ to obtain: \
\\
\begin{equation} 
\frac{d\psi}{dx} - if\psi =0 
\end{equation} 
\\
We then seek a nonvanishing function $g: M \rightarrow \mathbb{C}$ such multiplying both sides of $g$ will yield the equation $\frac{d}{dx}(g\psi) = 0$. Since $g(\frac{d\psi}{dx} -if) = g\frac{d\psi}{dx} - igf$ and $\frac{d(g\psi)}{dx} = g\frac{d\psi}{dx}+\frac{dg}{dx}\psi$, this implies that \\
\\
\begin{equation} 
\frac{dg}{dx} = -ifg
\end{equation} 
\\
which we can integrate using separation of variables to obtain \\
\\
\begin{equation} 
g(x) = \exp(-i\int_{0}^{x} f(t) \ dt) 
\end{equation}
\\
\textbf{Remark 1:} In the language of Dirac operators, equation (3) has the form \\
\begin{equation} 
\slashed{D}g = fgI
\end{equation} 
where $I$ is the identity endomorphsim on the spinor bundle. \\ 
\\
\textbf{Remark 2:} If $M^{1}$ is diffeomorphic to $\mathbb{R}$, then $g(x)$ will always be defined provided that $f(x)$ is continuous on $\mathbb{R}$. However, if $M^{1}$ is diffeomorphic to $S^{1}$, then $g(x)$ will only define a function on $S^{1}$ if $\int_{0}^{2\pi} f = n$ for $n \in \mathbb{Z}$, as this will allow $g$ to be periodic with period $2\pi$. Since a smooth, connected, one-dimensional manifold without boundary is diffeomorphic to either $S^{1}$ or $\mathbb{R}$ (cf. http://math.ucsd.edu/~lni/math250/1-manifold.pdf.  or p. 55 of [1]), we have the following theorem:
\begin{theorem} 
\noindent Let $(M^{1}, g)$ be a smooth, connected, one-dimensional manifold without boundary. Let $f: M^{1} \rightarrow \mathbb{R}$ be a continuous function. If $M^{1}$ is compact, then the Dirac-Schrodinger operator $\slashed{D} + fI$ admits a harmonic-type spinor if and only if $\int_{M^{1}} f = n$ for $n \in \mathbb{Z}$. If $M^{1}$ is non-compact, then the Dirac-Schrodinger operator $\slashed{D} + fI$ will always admit a harmonic-type spinor. 
\end{theorem}
\begin{proof} 
From the argument above, it is obvious that this reduces to solving for the complex valued function, $\psi$ that satisfies the equation $\frac{d}{dx}(g\psi) = 0$, where $g$ is the integrating factor defined in (22). The derivative of a function will be zero if and only if that function is constant. Therefore, $g\psi = C$ and $\psi = \frac{C}{g}$. The equation, (22), defining $g$ makes it clear that $g$ can never vanish. Therefore, provided $C$ is chosen to be nonzero, $\psi$ will be the desired nontrivial solution to the Dirac-Schrodinger equation $[\slashed{D} + fI]\psi =0$. Conversely, if $M^{1}=S^{1}$ and $\int_{M^{1}} f$ is \emph{not} an integer, then any solution of $[\slashed{D}+fI]\psi =0$ will still have the general solution $\psi = \frac{C}{g}$ for $0\leq x<2\pi$, but in this case $g(0)=1 \neq \lim_{x \rightarrow 2\pi} g(x)$, so that $\psi$ cannot be extended to a continuous function on $M^{1}$, let alone a differentiable one.
\end{proof}
\noindent Theorem 1 is then a direct corollary of theorem 8.
\subsection{Eigenvalues of Dirac-Schrodinger Operators in Dimension 1} 
\noindent Now, again assume that $M^{1}=S^{1}$ and that $\psi$ is an eigenspinor of $\slashed{D} + B$ with eigenvalue $\lambda$. Then it must satisfy the equation 
\begin{equation} 
i\frac{d}{dx}\psi + f\psi = \lambda \psi
\end{equation} 
Let $f_{avg}$ denote the average of $f$, defined by $f_{avg} =\frac{1}{|M^{1}|} \int_{M^{1}} f$. If we subtract this from both sides, then we obtain the eigenvalue equation: \\
\begin{equation} 
i\frac{d}{dx}\psi + (f-f_{avg})\psi = (\lambda - f_{avg})\psi
\end{equation} 
\\
Since $g$ never vanishes, $g\psi$ is an eigenvalue of $\slashed{D}$ with eigenvalue $\lambda-f_{avg}$. For simplicity, let $\phi = g\psi$. Then equation 8 becomes \\
\\
\begin{equation} 
i\frac{d\phi}{dx}=(\lambda- f_{avg})\phi
\end{equation}
Dividing both sides by $i \phi$ and integrating, we get \\
\\
\begin{equation} 
\ln (\phi) = -i(\lambda - f_{avg})x + C  
\end{equation}
or, solving for $\phi$, with $A=e^{C}$, 
\\
\begin{equation} 
\phi(x) = Ae^{-i(\lambda - f_{avg})x}
\end{equation} \\
\\
Now, in order for $\phi$ to be defined on $S^{1}$, we must have $\phi(0)=\phi(2\pi)$, or equivalently, $A=Ae^{-i(\lambda - f_{avg}(2\pi)}$. Dividing both sides by $A$, we get $e^{-i(\lambda - f_{avg})(2\pi)}=1$, from whence we can conclude that $\lambda-f_{avg}$ is an integer. We therefore have proven the following:

\begin{theorem} 
Let $M^{1}=S^{1}$ and $f$ is a smooth complex-valued function on $M^{1}$. Then the eigenvalues of $\slashed{D}$ are the integers $\{ n \}$ and the eigenvalues of $\slashed{D}+fI$ are the integers $n+f_{avg}$. If $M^{1} = \mathbb{R}$, then the eigenvalues of both $\slashed{D}$ and $\slashed{D}+fI$ consists of the set of all real numbers 
\end{theorem}
\section{Computation in Arbitrary Dimension $n$: Proof of Theorem 2}
In this section, we present the proof of theorem 2. From equation (3), it is clear that proving the existence of integrating factors reduces to computing $B=A^{-1}\slashed{D}A$ for a given invertible $A \in End(S)$. Let $f \in C^{\infty}(M^{n})$ and $X \in \Gamma(TM^{n})$; we proceed with the following computations: \\
\\
Case 1: If $A = exp(f)$, then $A^{-1} = exp(-f)$ and \begin{align*} 
\slashed{D}A & = c(e^{j})\nabla_{e_{j}}(exp(f)) \\
& = c(e^{j})exp(f)\nabla_{e_{j}}f \\
& = exp(f)c(\nabla f)
\end{align*} 
Therefore, we can conclude that \\
$$B=A^{-1}\slashed{D}A=c(\nabla f)$$\\
\\
Case 2: Let $\{ e_{j} \}$ be local orthonormal frame for which $\nabla_{e_{j}} e_{m} = 0$ at some $p \in M^{n}$.  Let  $X=X^{j}e_{j} \in \Gamma(TM^{n})$ and $A=c(X)=X^{j}c(e_{j})$. Note that since $c(X)^{2}=-g(X, X) = -|X|^{2}$, \\
\\
\begin{equation} 
A^{-1} = \frac{-c(X)}{|X|^{2}}
\end{equation} 
We then compute \\
\begin{align*} 
\slashed{D}A &=  c(e^{j})\nabla_{e_{j}}(X^{m}c(e_{n})) \\
&=c(e^{j})\nabla_{e_{j}}(X^{m})c(e_{m}) \\
&=\nabla_{e_{j}}X^{m}c(e^{j})c(e_{m}) \\
&=\Sigma_{m=1}^{n} \nabla_{e_{m}}X^{m}c(e_{m})^{2}+\Sigma_{j \neq m} \nabla_{e_{j}}X^{m}c(e^{j})c(e_{m}) \\
&= -\Sigma_{m=1}^{n} \nabla_{e_{m}}X^{m}+\Sigma_{j < m} \nabla_{e_{j}}X^{m}c(e^{j})c(e_{m})+\Sigma_{j>m}\nabla_{e_{j}}X^{m}c(e^{j})c(e_{m}) \\
&=-div \ X + \Sigma_{j>m}\nabla_{e_{j}}X^{m}c(e^{j})c(e_{m})-\Sigma_{m>j} \nabla_{e_{j}}X^{m}c(e_{m})c(e^{j}) \\
\end{align*} 
Since $j$ and $m$ are dummy indices, the last term in the last line of the above equation can be rewritten as $-\Sigma_{j>m}\nabla_{e_{m}}X^{j}c(e_{j})c(e_{j})$. It then follows that \\
\\
\begin{align*} 
\slashed{D}A &= -div \ X+\Sigma_{j>m}[\nabla_{e_{j}}X^{m}-\nabla_{e_{m}}X^{j}]c(e^{j})c(e_{m}) \\
&=-div \ X+ curl \ X
\end{align*}
where $curl \ X$ and $div \ X$ are defined by equations (5) and (6) respectively. It then follows that \\
\\
\begin{equation} 
B=A^{-1}\slashed{D}A=\frac{c(X)}{|X|^{2}}[(div \ X) - (curl \ X)]
\end{equation}
which is precisely form 2. This completes the proof of theorem 2. 
\section{Special Case: Proof of Theorem 3}
In the special case that $g=0$ and $\nabla_{e_{j}}f=0$ for all $j\geq 2$, we show that $A=exp(c(\nabla f))$ is an integrating factor for $B=\Delta_{g}f=-2\nabla_{e_{1}}\nabla_{e_{1}}f$ by computing $A^{-1}\slashed{D}A$. Note that first that $A^{-1}=exp(-c(\nabla f))$ as a consequence of the fact that $c(\nabla f)$, when restricted to each fiber of $S$ can be written as a matrix (cf Proposition 3 of [4]). Now, note that $c(\nabla f)^{2}=-|\nabla f|^{2}$, from whence it follows that \\
\\
$$
c(\nabla f)^{m} = \left\{
        \begin{array}{ll}
            (-1)^{m/2}|\nabla f|^{m} & \ $if$ \  m \  $is$ \ $even$\\
           (-1)^{(m-1)/2}|\nabla f|^{m-1}c(\nabla f) & $if$ \ m \ $is$ \ $odd$ \
        \end{array}
    \right.
$$
The following computations will show that, for all positive integers, $m$: \\
\\
\begin{equation} \slashed{D}c(\nabla f)^{m} = (\Delta_{g} f)(mc(\nabla f)^{m-1})
\end{equation} 
From equation (32), we can then conclude that \\
\\
\begin{align*}
\slashed{D}(exp(c(\nabla f)) &=\slashed{D}(\Sigma_{m=0}^{\infty} c(\nabla f)^{m}) \\
&=\Sigma_{m=0}^{\infty} (\frac{1}{m!})\slashed{D}(\frac{c(\nabla f)^{m}}{m!}) \\
&=(\Delta_{g}f)\Sigma_{m=1}^{\infty} \frac{m}{m!}(c(\nabla f)^{m-1})\\
&=(\Delta_{g}f)[\Sigma_{m=1}^{\infty}\frac{1}{(m-1)!}(c(\nabla f)^{m-1})] \\
&=(\Delta_{g} f)exp(c(\nabla f))
\end{align*}
It will then follow that 
\begin{align*} 
B &=A^{-1}\slashed{D}A \\
  &=exp(-c(\nabla f))(\Delta_{g} f)exp(c(\nabla f)) \\
  &=(\Delta_{g} f)(exp(-c(\nabla f))exp(c(\nabla f)) \\
  &=\Delta_{g} f
\end{align*}
which will conclude the proof of theorem 3. The computations for $m$ even and $m$ odd are given below: \\
\\
Case 1 ($m$ is even):
\begin{align*}
\slashed{D}(c(\nabla f)^{m})
&=(-1)^{m/2}c(e^{j})\nabla_{e_{j}}(|\nabla f|^{m}) \\
&=(-1)^{m/2}c(\nabla |\nabla f|^{m}) \\
&=(-1)^{m/2}c(\nabla ((|\nabla f|^{2})^{m/2})) \\
&=(-1)^{m/2}(m/2)(|\nabla f|^{2})^{(m-2)/2}c(\nabla (|\nabla f|^{2})
\end{align*}
Since $|\nabla f|^{2} = (\nabla_{e_{1}}f)^{2}$, we have \\
\\
\begin{equation}
\nabla |\nabla f|^{2} = 2(\nabla_{e_{1}}f)\nabla (\nabla_{e_{1}}f)
\end{equation}
Since $f$ is a function, $\nabla_{e_{j}}\nabla_{e_{1}} f=\nabla_{e_{1}}\nabla_{e_{j}}f$, so that $\nabla(\nabla_{e_{1}} f) = (\nabla_{e_{1}} \nabla_{e_{1}} f)e_{1}$. We can then simplify the above computations to the following if $m$ is even:\\
\begin{align*}
\slashed{D}(c(\nabla f)^{m}) &= (-1)^{m/2}(m/2)(|\nabla f|^{2})^{(m-2)/2})(2\nabla_{e_{1}}f)(\nabla_{e_{1}}\nabla_{e_{1}}f)c(e_{1})  \\
&=(-1)^{m/2}m|\nabla f|^{(m-2)}(\nabla_{e_{1}}f)(-\Delta_{g}f)c(e_{1})\\
&=(\Delta_{g} f)[m(-1)^{(m-2)/2}|\nabla f|^{(m-2)}c(\nabla f)] \\
&=(\Delta_{g}f)(mc(\nabla f)^{m-1})
\end{align*}

\noindent Case 2 ($m$ is odd): 
\begin{align*} 
\slashed{D}(c(\nabla f)^{m}) &=(-1)^{(m-1)/2}c(e^{j})\nabla_{e_{j}}(|\nabla f|^{m-1}c(\nabla f))\\
&=(-1)^{(m-1)/2}c(e^{j})[\nabla_{e_{j}}(|\nabla f|^{m-1})c(\nabla f) +|\nabla f|^{m-1}\nabla_{e_{j}}\nabla_{l}fc(e^{l})] \\
&=(\Delta_{g}f)((m-1)c(\nabla f)^{m-2}c(\nabla f)) +(-1)^{(m-1)/2}|\nabla f|^{m-1}c(e^{j})c(e^{l})\nabla_{e_{j}}\nabla_{e_{l}}) 
\end{align*} 
Since it was assume that $\nabla_{e_{j}} f=0$ for all $j \geq 2$, we see that the only nonzero terms of $c(e^{j})c(e^{l})\nabla_{e_{j}}\nabla_{e_{l}}f$ occur when $j=l=1$, leaving $c(e^{1})^{2}\nabla_{e_{1}}\nabla_{e_{1}}f$. Now, $c(e^{1})^{2}=-1$, so that $c(e^{1})^{2}\nabla_{e_{1}}\nabla_{e_{1}}f=- \nabla_{e_{1}}\nabla_{e_{1}} f = \Delta_{g} f$. \\
The above then becomes \\
\\
\begin{align*}
& ((\Delta_{g}f)((m-1)c(\nabla f)^{m-1}c(\nabla f)+(-1)^{(m-1)/2}|\nabla f|^{m-1}(\Delta_{g} f)) \\
&=(\Delta_{g}f)((m-1)c(\nabla f)^{m-1}+(-1)^{(m-1)/2}|\nabla f|^{m-1})
\end{align*}

\noindent Notice that because $m$ is odd, $c(\nabla f)^{m-1} = (-1)^{(m-1)/2}|\nabla f|^{m-1}$, so that we can conclude that \\
\\
\begin{equation} 
\slashed{D}(c(\nabla f)^{m})=(\Delta_{g}f)(mc(\nabla f)^{m-1})
\end{equation} 
This completes the proof of theorem 3. \\
\\
\\
\\
\section{Dimension 2: Proof of Theorem 4}
In this section, we investigate integrating factors for a Dirac-Schrodinger operator, $\slashed{D} + B$ on a surface (2 dimesnional Riemannian manifold) $M^{2}$. Before we obtain form 3 in theorem 3, we first check that $A^{-1}$ given by equation (13) is indeed in the inverse of the general endomorphism $A$ given in equation (12). In particular, we compute \\
\\
\begin{align*}
& \frac{f-c(X)-gc(e_{1}e_{2})}{f^{2}+g^{2}+|X|^{2}}[f+c(X)+gc(e_{1}e_{2})]\\ 
&=\frac{f^{2}-c(X)^{2}-g^{2}c(e_{1}e_{2})^{2}+fc(X)-fc(X)+fgc(e_{1}e_{2})-fgc(e_{1}e_{2})-gc(e_{1}e_{2})}{f^{2}+g^{2}+|X|^{2}} \\
&= \frac{f^{2}+|X|^{2}+g^{2}}{f^{2}+g^{2}+|X|^{2}}\\
&=1
\end{align*}
where the second to last line follows because $c(X)^{2}=c(X)c(X)=-|X|^{2}$ and $c(e_{1}e_{2})^{2}=c(e_{1})c(e_{2})c(e_{1})c(e_{2})=(-c(e_{1})^{2})(-c(e_{2}))^{2}=I^{2}I^{2}=I$, where $I$ is the identity endomorphism on $\Gamma(S)$. We next compute $\slashed{D}A$: \\
\\
\begin{align*}
\slashed{D}A &= c(e^{j})\nabla_{e_{j}}[f+c(X)+gc(e_{1}e_{2})] \\
&=c(e^{j})\nabla_{e_{j}}f+c(e^{j})\nabla_{e_{j}}X^{m}c(e_{m})+c(e^{j})\nabla_{e_{j}}gc(e_{1}e_{2}) \\
&=c(\nabla f)-div \ X+curl \ X+c(e^{1})\nabla_{e_{1}}gc(e_{1}e_{2}) \\
&=c(\nabla f)-div \ X+curl \ X+\nabla_{e_{1}}gc(e^{1})c(e_{1}e_{2})+\nabla_{e_{2}}gc(e^{2})c(e_{1}e_{2}) \\
&=c(\nabla f)-div \ X+curl \ X-(\nabla_{e_{1}}g)c(e_{2})+(\nabla_{e_{2}}g)c(e_{1})\\
&=-div \ X +[c(\nabla f)+(\nabla_{e_{2}}g)c(e_{1})-(\nabla_{e_{2}}g)c(e_{2})]+curl \ X
\end{align*} 
From this, it follows that \begin{align*}
B& =A^{-1}\slashed{D}A \\
& =\frac{f-c(X)-gc(e_{1}e_{2})}{f^{2}+g^{2}+|X|^{2}}[[-div \ X +[c(\nabla f)+(\nabla_{e_{2}}g)c(e_{1})-(\nabla_{e_{2}}g)c(e_{2})]+curl \ X]
\end{align*}
which is precisely form 4 given in equation 14. This proves theorem 4.  Note that theorem 5 was already proved by the discussion preceeding it in section 1.2.
\section{Proof of Positive Mass Theorems} 

\subsection{Definition of Asymptotic Flatness, ADM Energy, and Charge} 
Let $(M^{n}, g)$ be a Riemannian manifold of real dimension $n$ containing a vector field $E$ representing the Electric field. Let $K \subset M^{n}$ be a compact subset of $M^{n}$ and fix an arbitrary point $p \in K$. Below, $r$ will denote the variable, $d=d(x,p)$ for a variable point $x \in M^{n}$. For any number $t \in \mathbb{R}$, $B_{t}(0)$ will denote the subset of $\mathbb{R}^{n}$ consisting of all points in  whose distance from the origin ($0$) is not more than $t$. \\
\\
\textbf{Definition} $(M^{n}, g)$ is \emph{asymptotically flat} if there is a real number $r_{1} > 0$ and a diffeomorphism, $\Phi: \mathbb{R}^{n} - B_{r_{1}}(0) \rightarrow M^{n} - K$ such that the pull-back metric $\Phi^{*}g$ satisfies
\\
\begin{equation} 
|(\Phi^{*}g)_{ij} - \delta_{ij}| \leq \frac{C}{r}  
\end{equation} 
\\
and \\
\begin{equation} 
|(\Phi^{*}g)_{ij,k} | \leq \frac{C}{r^{2}}
\end{equation} 
\\
\noindent If $n=3$, $(M^{3}, g, E)$ will serve as initial data for the Einstein-Maxwell Equations in the time-symmetric case. An important global invariant of such a manifold is the ADM mass, representing the strength of the gravitational field at infinity, defined as in [8] (p. 373, equation 1.1.32) by \\
\\
\begin{equation} 
E_{ADM} = \lim_{r \rightarrow \infty} \frac{1}{16\pi} \int_{S(0, r)} (\partial_{l}g_{lj} - \partial_{j}g_{ll}) dS^{j}
\end{equation} 
where $S(0, r) = \partial (\Phi(\mathbb{R}^{n} - B_{r}(0))$ denotes the boundary of the complement of a coordinate ball of radius $r > r_{1}$ in $M^{n}$. In the following sections, we will also set $B(0, r) = \Phi(\mathbb{R}^{n} - B_{r}(0))$. \\
\\
Let $\nu$ denote a unit normal vector field on $S(0,r)$, pointing towards the asymptotically flat end of $M^{n}$. We then define the \emph{total charge}, Q, on $M^{n}$ by  \begin{equation}  
Q = \frac{1}{4\pi}\lim_{r \rightarrow \infty} \int_{S(0,r)} E(\nu)
\end{equation}   
To prove theorems 6-8, we define the charged Dirac Operator. This operator satisfies a Schrodinger-Lichnerowicz formula integral formula that will imply non-negativity of $E_{ADM} - |Q|$, or, in the case of theorem 7, $E_{ADM} - \sqrt{|P|_{g}^{2} + Q^{2}}$.  provided the integral over the boundary is non-negative. Non-negativity of this boundary integral will be ensured once a boundary condition is established for the spinor satisfying the Dirac equation.  
This will follow steps 1-3 of section 3 of [9]. 

\subsection{The Charged Dirac-Operator and The Schrodinger-Lichnerowicz Formulaes}
In this section, $D$ will denote the Levi-Cevita connection on $(M^{n}, g)$, defined by the Riemannian metric, $g$. $\{ e_{j}, j = 1, ..., n \} \in \Gamma(TM^{n})$ will represent an orthonormal frame for $TM^{n}$. For each, $j$, $B_{j}$ will denote a zeroth-order endormorphism of the spinor bundle; the kind we have been considering in previous sections. We will consider modified connections $\nabla$ on $TM^{n}$ locally as follows: \\
\\
\begin{equation} 
\nabla_{e_{j}} = D_{e_{j}} + B_{j}
\end{equation}
In this case, we can define a modified Dirac Operator, $D_{M}$ as follows: \\
\\
\begin{equation} 
\nabla_{e_{j}} = D_{e_{j}} + B_{j}
\end{equation}
In this case, we can define a modified Dirac Operator, $D_{M}$ as follows: \\
\\
\begin{equation} 
D_{M} = c(e^{j})\nabla_{e_{j}} = \slashed{D}+c(e^{j})B_{j}
\end{equation} 
Note that $D_{M}$ is precisely a Dirac-Schrodinger operator of the type considered above, with $B=c(e^{j})B_{j}$. Let $\phi \in \Gamma(S)$ be an arbitrary spinor field. Define $U_{\phi} \in \Gamma(TM^{n})$ locally by the following formula: \\
\\
\begin{equation} 
(U_{\phi})^{j} = <\phi, (\nabla_{e_{j}}\phi+c(e_{j})D_{M}\phi)>  
\end{equation} 
\noindent From now own, we will set $D_{j} = D_{e_{j}}$ and $\nabla_{j} = \nabla_{e_{j}}$ to simplify computation. The Dirac operator, $D_{M}$ given by (41) satisfies the following, \emph{Schrodinger-Lichnerowicz}  formula (cf. equation 3.3.7a-f of [8]):  \\
\\
\begin{equation}
D_{j}U^{j}  = |\nabla \phi|^{2} - |D_{M}\phi|^{2} +<\phi, (D_{j}D^{j}+\slashed{D})\phi> 
+<\phi, A\phi>
\end{equation}
\\
where $A \in End(S)$ is the endomorphism of $S$ given by formulae 3.3.7d-f of section 3.3 of [8]. Note that because of the formula 3.3.7d of [8], this endomorphism involves a covariant derivative and is therefore not zeroth order. 
\\
where $A \in End(S)$ is the endomorphism of $S$ given by formulae 3.3.7d-f of section 3.3 of [8]. Note that because of the formula 3.3.7d of [8], this endomorphism involves a covariant derivative and is therefore not zeroth order.   \\
\\
\noindent If an electric field $E \in \Gamma(TM^{n})$ is defined on $M^{n}$, then the connection on the spinor bundle is modified as follows: \\
\\
\begin{equation} 
\nabla_{j} = D_{j} - \frac{1}{2}c(E)c(e_{i})c(e_{0})
\end{equation}
\\
(cf. equation (57) of section 5.1 of [9]) \\
\\
In this case, the \emph{charged} Dirac-Operator, $D_{M}$, defined as $D_{M}=e^{j}\nabla_{e_{j}}$ takes the following form: \\
\\
\begin{equation} 
D_{M} =\slashed{D}-\frac{1}{2}c(E)c(e_{0})
 \end{equation} 
If, in addition, a symmetric two-tensor, $k \in \Gamma(TM^{n})$ is defined on $TM^{n}$, then the connection, $\nabla$ is modified as follows: \\
\\
\begin{equation} 
\nabla_{j} = D_{j}+ \frac{1}{2}k_{jl}c(e^{l})c(e_{0})-\frac{1}{2}c(E)c(e_{j})c(e_{0})
\end{equation} 
This creates the following formula for $D_{M}$: \\
\\
\begin{equation} 
D_{M} = \slashed{D}+\frac{1}{2}k_{jl}c(e^{j})c(e^{l})c(e_{0})-\frac{1}{2}c(E)c(e_{0})
\end{equation}
\\
\\
If we apply equation (41) to the operator given in (47), we obtain the following: \\
\\
\begin{align} 
D_{j}<\phi, \nabla^{j}+c(e_{j})D_{M}\phi> & = |\nabla \phi|^{2} - |D_{M}\phi|^{2} \\ 
& + \frac{1}{4}<\phi, (R+(tr_{g} k)^{2} -|k|_{g}^{2}-2|E|_{g}^{2})\phi \\
&+\frac{1}{4}<\phi, [2D_{j}(k^{lj}-(tr_{g}k)g^{lj})c(e^{l})c(e_{0})\phi> 
\end{align} 
We will define the energy density, $\mu$, the matter density vector field $J=J^{l}e_{l} \in \Gamma(TM^{3})$, and the ADM momentum vector field, $P=P^{l}e_{l} \in \Gamma(TM^{3})$ as follows: \\
\\
\begin{equation} 
\mu = R+(tr_{g} k)^{2} - |k|_{g}^{2} - 2|E|_{g}^{2} 
\end{equation}
\\
\begin{equation} 
J^{l} = 2D_{j}(K^{jl}-(tr_{g}k)g^{jl})
\end{equation} : \\
\\
\begin{equation} 
P^{j} = \frac{1}{4\pi}\lim_{r \rightarrow \infty} \frac{1}{2}\int_{S(0,r)} [(tr_{g} k)g^{jl}- k^{lj}]dS^{l}
\end{equation}
Then equations 48-50 can be rewritten as follows: \\
\\
\begin{align*} 
D_{j}<\phi, \nabla^{j}\phi+c(e_{j})D_{M}\phi> 
& = |\nabla \phi|^{2} - |D_{M}\phi|^{2} + 
\frac{1}{4}<\phi, [\mu +c(J)]\phi>
\end{align*} 
Now, assume that $M^{n}$ is an asymptotically flat manifold with boundary $\partial M^{n}$. Then we can apply the divergence theorem to both sides of the above equation to obtain the following: \\
\\
\begin{align*} 
\int_{B(0,r)} |\nabla \phi|^{2} - |D_{M} \phi|^{2} + \frac{1}{2}<\phi, [\mu + c(J)]\phi> & -\int_{\partial M^{n}} <\phi, \nabla^{j}\phi+c(e^{j})D_{M}\phi>dS^{j} \\
& = \int_{S(0,r)} <\phi, \nabla^{j}\phi + c(e_{j})D_{M}\phi>dS^{j}\\
\end{align*} 
\noindent Now, the convention is that $\phi$ asymptotes (as $r \rightarrow \infty$) to a constant spinor field, $\phi_{\infty}$. From equation (46), we calculate: \\
\begin{align} 
\nabla^{j}+c(e_{j})D_{M} &= \nabla^{j} + c(e^{j})c(e^{l})\nabla_{l} \\
&=D^{j}+\frac{1}{2}k_{jl}c(e^{l})c(e_{0})-\frac{1}{2}c(E)c(e_{j})c(e_{0}) \\ &+c(e^{j})c(e^{l})(D_{l}+\frac{1}{2}k_{lm}c(e^{m})c(e_{0}) -\frac{1}{2}c(E)c(e_{l})c(e_{0}))\\
& =D^{j}+c(e^{j})\slashed{D} \\
&+\frac{1}{2}k_{jl}c(e^{l})c(e_{0})+\frac{1}{2}k_{lm}c(e^{j})c(e^{l})c(e^{m})c(e_{0}) \\
& -\frac{1}{2}c(E)c(e_{j})c(e_{0})-\frac{1}{2}c(e^{j})c(e^{l})c(E)c(e_{l})c(e_{0})
\end{align} 
First, notice that $j$ is fixed and that $c(e^{l})c(e^{m})=-c(e^{m})c(e^{l})$, so that we can simplify (59) by first computing: 
\\
\begin{equation} 
k_{lm}c(e^{l})c(e^{m})c(e_{0})=(-\Sigma_{l=1}k_{ll})c(e_{0})=-(Tr_{g}k)c(e_{0})
\end{equation} 
This makes the terms in (59) reduce to \\
\\
\begin{equation} 
\frac{1}{2}(k_{jl}c(e^{l})-(Tr_{g}k)c(e^{j}))c(e_{0})
\end{equation} 
As in explained in ([8], p. 130) the expression (61) can be rewritten as follows: \\
\\
\begin{equation} 
\frac{1}{2}((Tr_{g}k)g^{jl}-k^{jl})c(e_{0})c(e_{l})
\end{equation} 
Now, to deal with the terms in (60), first note that any anti-symmetric terms will not contribute to the integral. By, anti-symmetric, we mean a term that becomes negative upon transposition; the transpose of $A \in End(S)$ will be denoted by $A^{T}$. The rules for transposition are as follows: $c(e_{0})^{T}=c(e_{0})$ and $c(e_{l})^{T}=-c(e_{l})$, provided that $l=1,...,n$. The transpose obeys the same rules for endomorphisms of the spinor bundle as it does for matrices. In particular, $(AB)^{T}=B^{T}A^{T}$. We therefore see in particular that $(c(e_{l})c(e_{0}))^{T} = c(e_{l})c(e_{0})$, so that $c(e_{l})c(e_{0})$ is symmetric while $[c(e_{m})c(e_{j})c(e_{0})]^{T}=c(e_{0})c(e_{j})c(e_{m})=c(e_{j})c(e_{m})c(e_{0})=-c(e_{m})c(e_{j})c(e_{0})$ is anti-symmetric if $j \neq m$.  It then follows that the only term of $-\frac{1}{2}c(E)c(e_{j})c(e_{0})$  that contributes  to the integral is $\frac{1}{2}E_{j}c(e_{0})$. Further, we can rewrite the term on the right-hand side of (60) by first rewriting $c(e^{l})c(E)c(e_{l})$ as follows:  \\
\\
\begin{equation} 
c(e^{l})c(E)c(e_{l})=c(e^{l})[-2g(E, e_{l})-c(e_{l})c(E)]=(n-2)c(E)
\end{equation} 
\\
\noindent Plugging this formula back into expression on the left hand side of (60), we obtain the following: 

\begin{equation} 
-\frac{(n-2)}{2}c(e^{j})c(E)c(e_{0})
\end{equation} 
As explained above, the only symmetric term, and thus the only term that contributes to the integral, in (65) is $\frac{(n-2)}{2}E_{j}c(e_{0})$. Therefore, the term from (60) contributing to the integral is $\frac{(n-1)}{2}E_{j}c(e_{0})$. We thus see that, as $r \rightarrow \infty$, the integral over $S(0,r)$ becomes the following: \\
: \\
\\
\begin{align} 
\lim_{r \rightarrow \infty} 
 & \int_{S(0,r)} <\phi_{\infty}, (D^{j}+c(e^{j})\slashed{D})\phi_{\infty}> dS^{j}\\
& \frac{1}{2}\int_{S(0,r)}  ((Tr_{g} k)g^{jl}-k^{jl})<\phi_{\infty}, c(e_{0})c(e_{j})\phi_{\infty}> dS^{j} \\
& + \int_{S(0,r)} <\phi_{\infty}, \frac{(n-1)}{2}E_{j}c(e_{0})\phi_{\infty}>dS^{j}
\end{align} 
It has been shown in [8] (cf. section 3.2) that (65) reduces to $4\pi E_{ADM}|\phi_{\infty}|^{2}$; from the definitions of $P$, $E_{ADM}$, and $Q$, (65)-(67) become the following:
\\
\begin{align}
4\pi <\phi_{\infty},[E_{ADM}+P^{j} c(e_{0})c(e_{j})+\frac{(n-1)}{2}Q]\phi_{\infty}>
\end{align}
Note that in the above expression, $I$ and there is a repeated sum over of $j$. \\
\\
\noindent Now, if $\phi$ is also a solution to the Dirac equation, $D_{M}\phi = 0$, then the equation at the top of page 19, upon letting $r \rightarrow \infty$ becomes the following: \\
\\
\begin{align} 
\int_{M^{n}} |\nabla \phi|^{2} & + \frac{1}{2}<\phi,  [\mu +c(J)]\phi> \\
&-\int_{\partial M^{n}}<\phi, \nabla^{j}\phi>dS^{j}\\
&=4\pi  <\phi_{\infty},[E_{ADM}+P^{j} c(e_{0})c(e_{j})+\frac{(n-1)}{2}Q]\phi_{\infty}>
\end{align}
\noindent Now, let $e_{j}, j = 1, ... , n$ be an orthonormal frame oriented so that $e_{1}$ is an outward pointing normal vector to $\partial M^{n}$ and the vector fields $e_{j}, \ j = 1,..., n-1$ are tangent to $\partial M^{n}$ will reduce to the following: \\
\\
\begin{equation} 
-\int_{\partial M^{n}} <\phi, \nabla^{1}\phi>
\end{equation} \\
\\
If $\phi$ is a solution of the Dirac equation, $c(e^{j})\nabla_{j}\phi = 0$, then we have the following: 
\\
\begin{equation} 
\nabla_{1}\phi=-\Sigma_{m=2}^{n}c(e_{1})c(e_{m})\nabla_{m}\phi]
\end{equation}
where there is a repeated sum over the index $m$. \\
\\
Using equation (46) for $\nabla_{m}$, we obtain the following: \\
\\
\begin{align} 
& \Sigma_{m=2}^{n} -c(e_{1})c(e^{m})\nabla_{m}\phi \\  &=\Sigma_{m=1}^{n-1} [-c(e_{1})c(e_{m})D_{m}\phi \\
&-\frac{k_{ml}}{2}c(e_{1})c(e_{m})c(e^{l})c(e_{0})\phi \\
&-\frac{1}{2}c(e_{1})c(e^{m})c(E)c(e_{m})c(e_{0})\phi
\end{align} 
Since $k$ is a symmetric tensor, the terms in (76) will cancel if unless $l=m$, leaving the following: \\
\\
\begin{equation} 
\Sigma_{m=2}^{n}-\frac{k_{mm}}{2}c(e_{1})c(e_{0})
\end{equation} \\
For (77), notice that \begin{align} c(e_{m})c(E)c(e_{m})  & = c(e_{m})[-2g(E,e_{m})-c(e_{m})c(E)] \\
& =(n-1)c(E)-2[proj_{\partial M^{n}}E] \\
&=(n-3)[proj_{\partial M^{n}} E]+(n-1)E^{1}c(e_{1})
\end{align}
Therefore, the term in (77) reduces to the following: \\
\\
\begin{equation} 
[\frac{n-3}{2}(proj_{\partial M^{n}}E)-\frac{n-1}{2}E^{1}]c(e_{0})
\end{equation} 
\subsection{Boundary Dirac-Schrodinger Operators}
It has been shown (cf. section 5.4 of [9]) that it is possible to define a boundary operator $\slashed{D}_{N_{j}}$ on each component $N_{j}$ of $\partial M^{n}$ and further that \\
\\
\begin{equation} 
\Sigma_{m=2}^{n} -c(e_{1})c(e^{m})\nabla_{m}\phi = \Sigma_{j} [-\slashed{D}_{N_{j}}+\frac{1}{2}H_{j}]
\end{equation}
For each $N_{j}$, define the following: \\
\\
\begin{equation} 
D_{N_{j}} = \slashed{D}_{N_{j}} -\frac{1}{2} \Sigma_{m=2}^{n} k_{mm}c(e_{1})c(e_{0}) + [\frac{n-3}{2}(proj_{\partial M^{n}}E)-\frac{n-1}{2}E^{1}]c(e_{0})
\end{equation} \\
Note that $D_{N_{j}}$ is a Dirac-Schrodinger operator on the component $N_{j}$ of $\partial M^{n}$ with \\
 \\
\\
\begin{equation}
B= [-\frac{1}{2} \Sigma_{m=2}^{n} k_{mm}c(e_{1}) + \frac{n-3}{2}(proj_{\partial M^{n}}E)-\frac{n-1}{2}E^{1}]c(e_{0})
\end{equation}
\\
Further (cf. section 6.1 of [9]), it can be shown that the eigenvalues of $D_{N_{j}}$ will remain the same if $B$ was replaced with the following: \\
\\
\begin{equation}
B= -\frac{1}{2} \Sigma_{m=2}^{n} k_{mm}c(e_{1}) + \frac{n-3}{2}(proj_{\partial M^{n}}E)-\frac{n-1}{2}E^{1}
\end{equation}
Then (74)-(77) become the following: \\
\\
\begin{equation} 
\Sigma_{j} [-D_{N_{j}}+\frac{1}{2}H_{j}]
 \end{equation} 
 \\
It then follows that (72) becomes the following 
\\
\begin{equation} 
\Sigma_{j} \int_{N_{j}} <\phi, [D_{N_{j}}-\frac{1}{2}H_{j}]\phi>
\end{equation} 
\subsection{Proof of Theorems 6-8} 
To prove theorems 6-8, we first note that we can solve $D_{M} \phi = 0$ subject to the boundary condition $P_{j}\phi = 0$, where $P_{j}$ denotes the projection onto the eigenspace of eigenspinors of $D_{N_{j}}$ that have corresponding negative eigenvalues (cf. section 7 of [9]). Let $\lambda_{1j}$ denote the smallest eigenvalue of $D_{N_{j}}$ and assume that $\frac{1}{2}H_{j} \leq \lambda_{j}$. Further, we can decompose $\phi$ into eigenspaces of $D_{N_{j}}$; $\phi = a^{lj}\phi_{lj}$ where $D_{N_{j}}\phi = \lambda_{lj}\phi$; we will assume that $\lambda_{l_{2}j} > \lambda_{l_{1}j}$ if $l_{2} > l_{1}$ (cf. section 5.4 of [9]; the paragraph after corollary 19). Following the same argument from section 5.4 of [9], we obtain: \\
\\
\begin{equation}
\Sigma_{j} \int_{N_{j}} <\phi, [D_{N_{j}}-\frac{1}{2}H_{j}]\phi> \geq \Sigma_{j} [\lambda_{nj}-\frac{1}{2}sup_{N_{j}} H_{j}](a^{lj})^{2}\int_{N_{j}} |\phi_{lj}|^{2}
\end{equation} 
\\
Now, if $\frac{1}{2}H_{j} \leq \lambda_{1j}$ for each $N_{j}$, then we can conclude that \\
\\
\begin{equation} 
\Sigma_{j} \int_{N_{j}} <\phi, [D_{N_{j}} - \frac{1}{2}H]\phi> \geq 0
\end{equation} 

\noindent Note that if $N_{j} = S^{2}$, then the Hijazi-Bar argument will give $\lambda_{j} = \sqrt{\frac{4\pi}{A_{j}}}$, where $A_{j}$ denotes the area of $N_{j}$ will respect to the induced Riemannian metric. \\
\\
From equations (69)-(71), we thus have the following: \\
\\
\begin{align}
4\pi<\phi_{\infty}, [E_{ADM}+P^{j}c(e_{0})c(e_{j})+& \frac{n-1}{2}Q]\phi_{\infty}> \\
& \geq \int_{M^{n}} |\nabla \phi|^{2} +\frac{1}{4}<\phi, [\mu + c(J)]\phi>
\end{align}
\noindent Now, if $\mu \geq |J|_{g}$ (the dominant energy condition), then the term (92) is nonnegative; an argument from [8], section 3.3 will then imply that $E_{ADM} \geq \sqrt{|P|_{g}^{2} + Q^{2}}$, because the choice of the constant $\phi_{\infty}$ spinor field is arbitrary. Now, case 1 of theorem 6 is simply theorem 2 of [9]. Case 2 follows because the first eigenvalue of $D_{N_{j}}$ will be $\sqrt{4\pi}{A_{j}}-\frac{4|Q_{j}|}{A_{j}}$ if $E(\nu)$ is only dependant on a single variable by theorem 5, provided that 
$\sqrt{\frac{4\pi}{A_{j}}} - \frac{4|Q_{j}|}{A_{j}} > 0$. For theorem 7, case 1 and case 3 follow from the Hijazi-Bar argument and case 2 follows from the following argument. If $B-\frac{4\pi Q_{j}}{A_{j}}$ admits an integrating factor, then the Hijazi-Bar argument will yield $|\lambda_{j} - \frac{4\pi Q_{j}}{A_{j}}| \geq \sqrt{\frac{4\pi}{A_{j}}}$; assuming that $Q_{j} <0$, this  yields $\lambda_{1j} \geq \sqrt{\frac{4\pi}{A_{j}}}-\frac{4\pi|Q_{j}|}{A_{j}}$. Theorem 8 follows simply by replacing $\sqrt{\frac{4\pi}{A_{j}}}$ with $\lambda_{j}$ (the smallest eigenvalue of $\slashed{D}_{N_{j}}$. Note: If $H \leq 0$, then an argument from a paper by Hawking-Horowitz and Perry [10], shows that the lower bound on $E_{ADM}$ holds. This concludes the proof of theorems 6-8. 
\section{Outlook} 
In theorems 6-8, we have not considered the case of equality. That is, if $E_{ADM} = |Q|$ or if  $E_{ADM} = \sqrt{|P|_{g}^{2}+Q^{2}}$ on $(M^{n}, g)$, what would that necessarily say about $(M^{n}, g, E)$ and $H_{j}$, or, in the case of theorems 7 and 8, about $k_{mm}$?


\begin{thebibliography}{999}



\bibitem{M. Herzlich} 
Herzlich, Marc
\emph{A Penrose-like Inequality for the Mass of Riemannian Asymptotically Flat Manifolds}
Communications in Mathematical Physics 
September 1997, Volume 188, Issue 1, pp. 121-133
\bibitem{Ginoux}

N.G. Ginoux
The Dirac Spectrum.
Springer-Verlag Berlin Heidelberg 2009

\bibitem{Hijazi} 

O. Hijazi 
A Conformal Lower Bound for the Smallest Eigenvalue of the Dirac Operator and Killing Spinors 
Communications in Mathematical Physics, Volume 104 pp. 151-162 (1986)

\bibitem{Klain}
  Klain, Dan
  \emph{The Matrix Exponential (with exercises)}.\\
  http://faculty.uml.edu/dklain/exponential.pdf, 2018

\bibitem{B. Morel} 
Morel, Bertrand \emph{Eigenvalue Estimates for the Dirac-Schrodinger Operators} 
Journal of Geometry and Physics 38(1): 1-18 April 2001 
\bibitem{Parker and Taubes} 

T. Parker and C.H. Taubes 
On Witten's Proof of the Positive Energy Theorem 
Communications in Mathematical Physics, Volume 84, 223-238 (1982)

\bibitem{Gibbons and Hull} 
C.W. Gibbons and C.M. Hull
A Bogolomny Bound for General Relativity and Solitons in N = 2 Supergravity 
Physics Letters, Volume 109B, number 3 (1982)



\bibitem{Chrusciel}
P.T. Chrusciel
Lectures on Energy in General Relativity
University of Vienna (piotr.chrusciel@univie.ac.at) February 22, 2013 

\bibitem{Abramovic}
R. Abramovic
\emph{The Positive Mass Theorem with Charge Outside Horizon(s)} 
http://www.math.stonybrook.edu/alumni/2018-Robert-Abramovic.pdf


\bibitem{Hawking}
G.W. Gibbons, S.W. Hawking, G.T. Horowitz, M.P. Perry \emph{Positive Mass Theorems for Black Holes} Communications in Mathematical Physics, Volume 88, pp. 295-308 (1983)
\end{thebibliography}
\end{document}